\documentclass[a4paper,UKenglish,cleveref, autoref, thm-restate]{lipics-v2019}

\graphicspath{{./figures/}}

\hideLIPIcs

\usepackage{microtype} 
\usepackage{amsmath}
\usepackage{amssymb}
\usepackage{subcaption}
\usepackage{xspace}
\usepackage{tabu}
\usepackage{amsmath}
\usepackage{booktabs}
\usepackage{wrapfig}
\usepackage{amsbsy}
\usepackage{placeins}
\usepackage{xcolor}

\modulolinenumbers[5]

\newcommand{\crnum}{\textnormal{cr}}
\newcommand{\Gell}{\ensuremath{G_\ell}}
\newcommand{\edge}[2]{\ensuremath{#1#2}}
\newcommand{\kplanarsub}[1]{{#1\mathrm{-pl}}}
\newcommand{\quasisub}{{\mathrm{quasi}}}

\newcommand{\fanplsub}{{\mathrm{fan}}}
\newcommand{\fanfreesub}{{\mathrm{fan-free}}}
\newcommand{\kfanfreesub}[1]{{#1\mathrm{-fan-free}}}
\newcommand{\plconsub}{{\mathrm{pl-con}}}
\newcommand{\kskewsub}[1]{{\mathrm{skew-}#1}}
\newcommand{\kapexsub}[1]{{#1\mathrm{-apex}}}
\newcommand{\kgapsub}[1]{{#1\mathrm{-gap}}}
\newcommand{\klgridsub}[2]{{#1,#2\mathrm{-grid-free}}}
\newcommand{\racsub}{{\mathrm{RAC}}}
\newcommand{\etal}{\textit{et al.}\xspace}
\newcommand{\mypar}[1]{\smallskip\noindent{\sffamily\bfseries\boldmath#1}}

\title{Crossing Numbers of Beyond-Planar Graphs Revisited}

\titlerunning{Crossing Numbers of Beyond-Planar Graphs Revisited}

\author{Nathan van Beusekom, Irene Parada, and Bettina Speckmann}{TU Eindhoven, the Netherlands}{[n.a.c.v.beusekom|i.m.de.parada.munoz|b.speckmann]@tue.nl}{}{}

\authorrunning{Nathan van Beusekom, Irene Parada, and  Bettina Speckmann}

\Copyright{John Q. Public and Joan R. Public} 


\category{} 

\relatedversion{} 

\supplement{}

\nolinenumbers 

\EventEditors{John Q. Open and Joan R. Access}
\EventNoEds{2}
\EventLongTitle{42nd Conference on Very Important Topics (CVIT 2016)}
\EventShortTitle{CVIT 2016}
\EventAcronym{CVIT}
\EventYear{2016}
\EventDate{December 24--27, 2016}
\EventLocation{Little Whinging, United Kingdom}
\EventLogo{}
\SeriesVolume{42}
\ArticleNo{23}

\begin{document}

\maketitle

\begin{abstract}
Graph drawing beyond planarity focuses on drawings of high visual quality for non-planar graphs which are characterized by certain forbidden edge configurations. A natural criterion for the quality of a drawing is the number of edge crossings. The question then arises whether beyond-planar drawings have a significantly larger crossing number than unrestricted drawings. Chimani et al. [GD'19] gave bounds for the ratio between the crossing number of three classes of beyond-planar graphs and the unrestricted crossing number. In this paper we extend their results to the main currently known classes of beyond-planar graphs characterized by forbidden edge configurations and answer several of their open questions.
\end{abstract}

\section{Introduction}
\label{sec:introduction}

A central topic in graph drawing are high quality drawings which are not necessarily planar.
A natural criterion for the quality of a drawing is the number of edge crossings.
While empirical studies suggest that the number of crossings is not the only factor that influences human understanding of a drawing, it is nevertheless one of the most relevant aesthetic indicators~\cite{aestheticsGD14,aestheticGD97,aesthetic97,aesthetics02}.
However, it is NP-hard to compute the minimum number of crossings of a graph~$G$ across all possible drawings of $G$~\cite{gj_cnnp1983}.
This minimum number of crossings is also referred to as the \emph{crossing number} of $G$, denoted by $\crnum(G)$. Meaningful upper and lower bounds, as well as heuristic, approximation, and parameterized algorithms, have hence been a major focus~\cite{crossingsGDhandbook13}.

Drawing edges as straight-line segments arguably simplifies any drawing.
One classic variant is hence the \emph{rectilinear crossing number} $\overline{\crnum}(G)$ of a graph $G$, that is, the minimum number of crossings across all straight-line drawings of~$G$.
Clearly $\crnum(G) \leq \overline{\crnum}(G)$ for any graph~$G$.
The straight-line restriction increases the crossing number arbitrarily: Bienstock and Dean~\cite{cn-rcn-far-93} showed that
for every $k$ there exists a graph $G$ such that $\crnum(G) = 4$
and  $\overline{\crnum}(G) = k$.
On the other hand, if $G=(V,E)$ is a graph with maximum degree $\Delta$ then
$\mathrm{\overline{cr}}(G) = O(\Delta\cdot \mathrm{cr}^2(G))$~\cite{cn-rcn-92,cn-rcn-95} and when $|E|\ge 4|V|$ this bound can be improved to  $\mathrm{\overline{cr}}(G) = O(\Delta \mathrm{cr}(G) \log \mathrm{cr}(G))$; see Schaefer's book~\cite{bookcn} on crossing numbers.
Computing $\mathrm{\overline{cr}}(G)$ is also NP-hard (actually, it is $\exists \mathbb{R}$-complete~\cite{ERcn91,ER_GD09})
but in polynomial time it can be approximated to $\mathrm{\overline{cr}}(G)+ o(n^4)$~\cite{approx-rcn-19}.

In recent years there has been particular interest in drawings of \emph{beyond-planar} graphs, which are characterized by certain forbidden crossing configurations of the edges; see the recent survey by Didimo~\etal~\cite{Didimo_2019}.
In this paper we study the relation between the
crossing number restricted to beyond-planar drawings,
and the (unrestricted) crossing number.
The \emph{restricted crossing number} $\crnum_{\mathcal{F}}(G)$ of a graph $G$ is the minimum number of crossings required to draw $G$ such that the drawing belongs to the beyond-planar family~$\mathcal{F}$.
For a beyond-planar family $\mathcal{F}$, the
\emph{crossing ratio}~$\varrho_\mathcal{F}$ is
defined as $\varrho_\mathcal{F}= \sup_{G\in \mathcal{F}} \crnum_\mathcal{F}(G)/\crnum(G)$, that is,
the supremum over all graphs in $\mathcal{F}$ of the ratio between the restricted crossing number and the (unrestricted) crossing number.

Very recently, Urschel and Wellens~\cite{Urschel2021} studied a related concept specifically for $k$-planar graphs. The \emph{local crossing number} $\textnormal{lcr}(\Gamma)$ of a drawing $\Gamma$ of graph $G$ is the maximum number of crossings per edge. Correspondingly, the local crossing number $\textnormal{lcr}(G)$ of $G$ is the minimum over all possible drawings of $G$. Hence a graph $G$ is $k$-planar if its local crossing number is $k$. Urschel and Wellens aim to find drawings that simultaneously minimize the crossing number and the local crossing number, that is, they aim to minimize
$\frac{\crnum(\Gamma)}{\crnum(G)} \cdot \frac{\textnormal{lcr}(\Gamma)}{\textnormal{lcr}(G)}$.

\begin{table}[b]
\caption{Bounds for the supremum $\varrho$ of the ratio between crossing numbers, for our upper bounds we assume the drawings to be \emph{simple}, that is, two edges share at most one point.}
\label{tab:results}
\small
\begin{tabu} to \textwidth {X[1,l,m]|X[1.67,c,m]X[1,c,m]|X[1,c,m]|X[1,c,m]}
        \toprule
        {\bfseries Family} &
        \multicolumn{2}{c|}{\bfseries Forbidden Configurations} &
        {\bfseries Lower} &
        {\bfseries Upper} \\
        \midrule
        $k$-planar & An edge crossed more than $k$ times & \includegraphics[page=1]{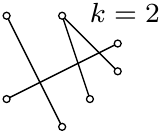} & $\boldsymbol{\Omega (n/k)}$ \newline (Section~\ref{sec:k-planar})& $\boldsymbol{O(k\sqrt{k}n)}$ \newline (Section~\ref{sec:k-planar}) \\
        \midrule
        $k$-quasi-planar & $k$ pairwise crossing edges & \includegraphics[page=2]{forbidden-configs} & $\Omega(n/k^3)$ \newline \cite{chimani} & $f(k)n^2\log ^2 n$ \newline \cite{chimani} \\
        \midrule
        Fan-planar & Two independent edges crossing a third or two adjacent edges crossing another edge from different "sides" & \includegraphics[page=7]{forbidden-configs} & $\Omega(n)$ \cite{chimani} & $O(n^2)$ \cite{chimani} \\
        \midrule
        $(k,l)$-grid-free & Set of $k$ edges such that each edge crosses each edge from a set of $l$ edges. & \includegraphics[page=5]{forbidden-configs} & $\Omega\Big(\dfrac{n}{kl(k+l)}\Big)$  \newline (Section~\ref{sec:grid-free})& $g(k,l)n^2$ \newline (Section~\ref{sec:grid-free}) \\
        \midrule
        $k$-gap-planar & More than $k$ crossings mapped to an edge in an optimal mapping &  \includegraphics[page=8]{forbidden-configs} & $\boldsymbol{\Omega(n/k^3)}$ \newline (Section~\ref{sec:gap-planar})& $\boldsymbol{O(k\sqrt{k}n)}$ \newline (Section~\ref{sec:gap-planar}) \\
        \midrule
        Skewness-$k$ & Set of crossings not covered by at most $k$ edges & \includegraphics[page=3]{forbidden-configs} & $\boldsymbol{\Omega(n/k)}$ \newline (Section~\ref{sec:skewness-k})& $\boldsymbol{O(kn + k^2)}$ \newline (Section~\ref{sec:skewness-k}) \\
        \midrule
        $k$-apex & Set of crossings not covered by at most $k$ vertices & \includegraphics[page=4]{forbidden-configs} & $\Omega(n/k)$ \newline (Section~\ref{sec:k-apex})& $O(k^2n^2 + k^4)$ \newline (Section~\ref{sec:k-apex}) \\
        \midrule
        Planarly connected & Two crossing edges that do not have two of their endpoint connected by a crossing-free edge & \includegraphics[page=9]{forbidden-configs} & \hspace{.8em}$\boldsymbol{\Omega(n^2)}$ \newline (Section~\ref{sec:plan-con})& \hspace{.8em}$\boldsymbol{O(n^2)}$ \newline (Section~\ref{sec:plan-con}) \\
        \midrule
        $k$-fan-crossing-free & An edge that crosses $k$ adjacent edges & \includegraphics[page=6]{forbidden-configs} & $\boldsymbol{\Omega(n^2/k^3)}$ \newline (Section~\ref{sec:fan-free})& $\boldsymbol{O(k^2n^2)}$ \newline (Section~\ref{sec:fan-free}) \\
        \midrule
        Straight-line~RAC & Two edges crossing at an angle $<\frac{\pi}{2}$ & \includegraphics[page=10]{forbidden-configs} & \hspace{.8em}$\boldsymbol{\Omega(n^2)}$ \newline (Section~\ref{sec:rac})& \hspace{.8em}$\boldsymbol{O(n^2)}$ \newline (Section~\ref{sec:rac}) \\
        \bottomrule
\end{tabu}
\end{table}

\mypar{Results.}
Chimani~\etal~\cite{chimani} gave bounds on the crossing ratio for the 1-planar, quasi-planar, and fan-planar families.
These results show that there exist graphs that have significantly larger crossing numbers when drawn with the beyond-planar restrictions. Their 1-planarity bound also applies to $k$-planarity when allowing parallel edges. We show that parallel edges are not needed when extending their proof to $k$-planarity. To do so, we introduce the concept of \emph{$k$-planar compound edges}, which exhibit essentially the same behavior as $k$ parallel edges. See Section~\ref{sec:k-planar} for details.
In Section~\ref{sec:grid-and-gap} we show how to extend the proof constructions of Chimani~\etal for quasi-planar and fan-planar graphs to two additional classes of beyond-planar graphs. These constructions use the concept of $\ell$-compound edges; in Section~\ref{sec:skew-apex} we show how to use such $\ell$-compound edges to prove lower bounds for two further families of beyond-planar graphs.
Finally, in Section~\ref{sec:plan-con-fan-free} we introduce the concept of \emph{$\ell$-bundles} which allow us to prove tight bounds on the crossing ratio for families of graphs where these bounds are quadratic in the number $n$ of vertices.
All bounds are summarized in Table~\ref{tab:results}. Bounds that are tight for a fixed~$k$ are indicated in boldface.
Last but not least, in Section~\ref{sec:straight-lines} we show that all bounds also apply to straight-line drawings.

\section{\boldmath $k$-Planar Graphs}
\label{sec:k-planar}

In a $k$-planar drawing no edge can be involved in more than $k$ crossings. Chimani~\etal~\cite{chimani} prove a tight bound on the crossing ratio for 1-planarity: $\varrho_\kplanarsub{1} = n/2-1$.
They also noted that the same arguments hold for $k$-planarity, if one allows parallel edges (multiple edges between two vertices).
They achieve a $k(n-2)/2$ lower bound on $\varrho_\kplanarsub{k}$ for graphs with $n$ vertices
by replacing all edges except for one in their construction for $1$-planarity by~$k$~parallel edges.
We show that parallel edges are not needed to extend their proof to $k$-planarity.
However, the dependence on $k$ in the lower bound that we achieve is worse than the lower bound  by Chimani~\etal using parallel edges.

To extend their construction  to $k$-planarity we introduce \emph{$k$-planar compound edges} (see Figure~\ref{fig:k-planar-compound-edges-appendix}), which exhibit essentially the same behavior as $k$ parallel edges.
One $k$-planar compound edge consists of $k^2$ parallel edges, each subdivided $k-1$ times (so it consists of a total of $k^3$ edges).
We replace each set of parallel edges in their construction with a $k$-planar compound edge. In a $k$-planar setting, each $k$-planar compound edge can cross one other $k$-planar compound edge such that each edge has exactly $k$ crossings (see Figure~\ref{fig:k-planar-compound-edges-appendix}).

\begin{figure}[b]
    \centering
    \includegraphics[page=4]{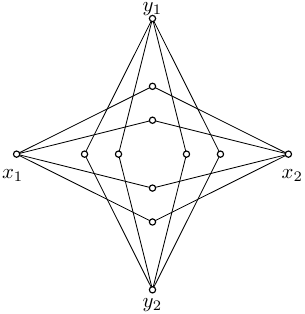}
    \hskip 2cm
    \includegraphics[page=5]{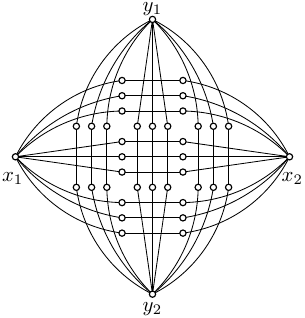}
    \caption{Left: 2-planar compound edges crossing; Right: 3-planar compound edges crossing}
    \label{fig:k-planar-compound-edges-appendix}
\end{figure}
\begin{theorem}\label{theorem:k-planar}
For every $\ell \geq 7$ there exists a $k$-planar graph $G_\ell^k$
with $n= 23\ell (k-1)k^2 + 11\ell + 2$ vertices
such that
$\crnum(\Gell^k) \leq 2k^2$ and $\crnum_\kplanarsub{k}(\Gell^k) \geq k^4(11\ell)$, thus $\varrho_\kplanarsub{k} \in \Omega(n/k)$.
\end{theorem}

\begin{figure}[h]
    \centering
    \includegraphics[page=6]{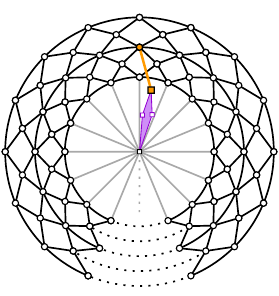}
    \hskip 2cm
    \includegraphics[page=7]{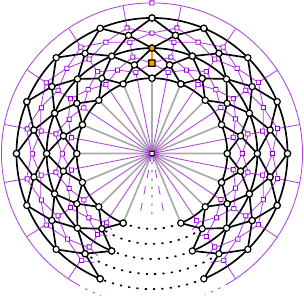}
    \caption{Construction in~\cite{chimani}. Left: drawing of $\Gell$ with $\crnum(\Gell) = 2$; Right: 1-planar drawing of $\Gell$ with $\crnum_\kplanarsub{1}(\Gell) = 11\ell$.}
    \label{fig:1-planar-chimani}
\end{figure}

\begin{proof}
We construct the graph $\Gell$ as described by Chimani \etal~\cite{chimani}, Section 2 (see Figure~\ref{fig:1-planar-chimani}). The construction of $\Gell$ consists of three parts: a rigid graph $P$ (round vertices, black edges); its dual $P^*$ (square vertices, purple); a set of \emph{binding edges} which connects $P$ to $P^*$ (grey);  and one \emph{special edge} (orange).
The graph $\Gell$ has $11\ell + 2$ vertices and $23\ell + 1$ edges.

Subsequently, we replace each edge in $\Gell$ except the special edge with a $k$-planar compound edge. That is, we replace every black, grey, and purple edge by $k^2$ parallel edges, all split by $k-1$ vertices. The resulting graph $\Gell^k$ has $n= 23\ell (k-1)k^2 + 11\ell + 2$ vertices.
Two $k$-planar compound edges may cross while maintaining $k$-planarity (see Figure~\ref{fig:k-planar-compound-edges-appendix}). The graph~$\Gell^k$ admits a drawing with $2k^2$ crossings (see Figure~\ref{fig:1-planar-chimani}, left) where $P^*$ is drawn entirely in one face (shaded purple). Thus, $\crnum(\Gell^k) \leq 2k^2$. Note that this drawing is not $k$-planar.

We show that $\crnum_\kplanarsub{k}(\Gell^k) \geq k^4(11\ell)$.
Suppose we have a $k$-planar, crossing minimal drawing~$\Gamma$ of~$\Gell^k$.
We remove the interior vertices of each $k$-planar compound edge in~$\Gamma$ and hence each path of length $k$ becomes one edge. Thus, each $k$-planar compound edge becomes $k^2$ parallel edges, which results in a drawing $\Gamma'$ of a multigraph~$G_{\ell,k^2}$.
Since we removed only subdivisions of edges, 
it follows that the number of crossings in~$\Gamma'$ is equal to the number of crossings in~$\Gamma$. Observe that $G_{\ell,k^2}$ is a modification of $\Gell$ where each edge except the special orange edge is replaced with $k^2$ parallel edges.
This multigraph $G_{\ell,k^2}$ with $11\ell + 2$ vertices is the same graph as the one used in Corollary~3~\cite{chimani} which proves
that $\crnum_\kplanarsub{k^2}(G_{\ell,k^2}) = k^4(11\ell)$.

In $\Gamma$ each edge in each path of length $k$ had at most $k$ crossings.
Thus, in $\Gamma'$ each edge has at most~$k^2$ crossings (the special edge has at most $k$ crossings). Therefore, $\Gamma'$ is a $k^2$-plane drawing of~$G_{\ell,k^2}$.
Thus, there are at least $k^4(11\ell)$ crossings in $\Gamma'$
and at least $k^4(11\ell)$ crossings in $\Gamma$.
Since $\Gamma$ is a crossing minimal drawing of $\Gell^k$ (by definition), it follows that  $\crnum_\kplanarsub{k}(\Gell^k) \geq k^4(11\ell)$.
Recall that $\crnum(\Gell^k) \leq 2k^2$ and $n= 23\ell (k-1)k^2 + 11\ell + 2$.
We can conclude that $\varrho_\kplanarsub{k} \in \Omega(n/k)$.
\end{proof}
A $k$-planar graph with $n$ vertices can have at most $3.81\sqrt{k}n$ edges~\cite{maxedgenr_k-pl}. Each edge can have at most $k$ crossings. Thus, a crossing-minimal $k$-plane drawing cannot have more than $k \cdot 3.81\sqrt{k}n$ crossings and therefore $\varrho_\kplanarsub{k} \in O(k\sqrt{k}n)$.

\section{\boldmath $(k,l)$-Grid-Free and $k$-Gap-Planar}
\label{sec:grid-and-gap}

\begin{wrapfigure}[5]{r}{0.18\linewidth}
  \vspace{-3\baselineskip}
  \raggedleft
  \includegraphics[page=1]{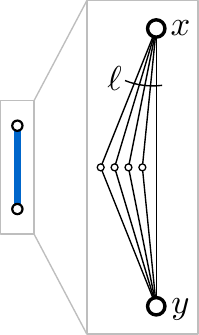}
  \label{fig:ellcompound-edge}
\end{wrapfigure}
Chimani~\etal~\cite{chimani} prove bounds on the crossing ratio for $k$-quasi-planarity and
fan-planarity using the notion of extended edges, which we refer to as \emph{$\ell$-compound edges}: an $\ell$-compound edge $\edge{x}{y}$ is a combination of the edge $\edge{x}{y}$ and a set $\Pi_{xy}$ of $\ell-1$ paths of length two connecting $x$ and $y$ (see figure on the right). We indicate $\ell$-compound edges with thick blue lines.

In their Lemma~5 (restated as Lemma~\ref{lemma:ellcompound} below), Chimani~\etal show that if two $\ell$-compound edges cross in a drawing $\Gamma$, then $\Gamma$ contains at least $\ell$ crossings. Hence, for a drawing to have less than $\ell$ crossings no two $\ell$-compound edges can cross. Chimani~\etal then construct graphs where the endpoints of an $\ell$-compound edge are separated by a closed curve which forces $\ell$ crossings. They conjectured that similar constructions could be used on additional families of beyond-planar graphs. Here we show that this is indeed the case by using similar constructions to prove lower bounds for $(k,l)$-grid-free graphs and $k$-gap-planar graphs.

\vspace{1.25\baselineskip}
\begin{lemma}[Lemma~5~in~Chimani~\etal~2019~\cite{chimani}]
\label{lemma:ellcompound}
Let $G$ be a graph containing two independent edges $\edge{u}{v}$ and $\edge{w}{z}$. Suppose that $u$ and $v$ ($w$ and $z$, resp.) are connected by a set $\Pi_{uv}$ ($\Pi_{wz}$, resp.) of $\ell-1$ paths of length two. Let $\Gamma$ be a drawing of $G$. If $\edge{u}{v}$ and $\edge{w}{z}$ cross in $\Gamma$, then $\Gamma$ contains at least $\ell$ crossings.
\end{lemma}
For both classes we prove lower bounds on the crossing ratio which are linear in the number of vertices~$n$. Specifically, in Section~\ref{sec:grid-free}, we prove a general lower bound for the $(k,l)$-grid-free crossing ratio for all $k>1$ and $l>1$. Then, in Section~\ref{sec:gap-planar}, we first prove a tight lower bound for 1-gap-planarity and then extend it to $k$-gap-planarity.

\subsection{\boldmath $(k,l)$-Grid-Free Graphs}
\label{sec:grid-free}

In a $(k, l)$-grid-free drawing it is forbidden to have a $k \times l$ grid, that is, a set of $k$ edges that cross a set of $l$ edges. We create a cycle of length $2k+2l$ of $\ell$-compound edges and connect each vertex of the cycle to an additional vertex $x$ using a further $2k+2l$ $\ell$-compound edges. By Lemma~\ref{lemma:ellcompound} none of the $\ell$-compound edges can cross or the drawing has $\ell$ crossings. Note that $k$ and $l$ are constants, while $\ell$ can become arbitrarily large with respect to $k$ and $l$. We add a $k \times l$ grid connecting the vertices of the cycle (see Figure~\ref{fig:kl-grid-free-appendix}, left; $\ell$-compound edges are indicated in blue, grid edges in black). In an $(k, l)$-grid-free drawing, one grid edge has to be drawn in the same face (with respect to the cycle) as the vertex $x$ and must hence cross an $\ell$-compound edge, incurring $\ell$ crossings.

\begin{figure}[b]
    \centering
    \includegraphics[page=4]{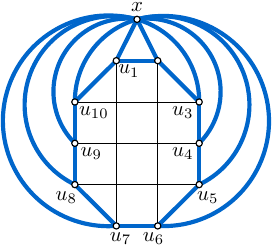}
    \hskip 2cm
    \includegraphics[page=5]{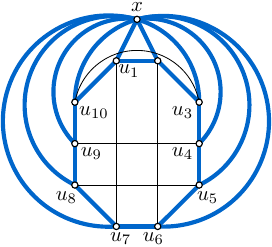}
    \caption{Left: drawing of $\Gell$ with $\crnum(G_\ell) \leq 6$; Right: $(2,3)$-grid-free drawing of $\Gell$ with $\crnum_\klgridsub{2}{3}(\Gell) \geq \ell$, $\ell$-compound edges drawn with thick blue lines.}
    \label{fig:kl-grid-free-appendix}
\end{figure}

\begin{theorem}\label{theorem:kl}
For every $\ell \geq 2, k > 1, l > 1$ there exists a $(k,l)$-grid free graph $G_\ell$ with $n = 4(k+l)\ell - 2(k+l) + 1$ vertices such that $\crnum_\klgridsub{k}{l}(G_\ell) \geq \ell$ and $\crnum(G_\ell) \leq kl$.
\newline Thus, $\varrho_\klgridsub{k}{l} \in \Omega\Big(\dfrac{n}{kl(k+l)}\Big)$.
\end{theorem}
\begin{proof}
We first construct a graph $G'$ which consists of a cycle $C$ of length $2k + 2l$, $C = \langle u_1,\ldots,u_{2k+2l}\rangle$ and a vertex $x$ connected to each vertex of~$C$. Now we construct a graph $\Gell$ by replacing each edge in $G'$ by an $\ell$-compound edge. Finally, we add the grid edges $\edge{u_i}{u_{2k+l-i+1}}$ for $i=1,\ldots,k$ and $\edge{u_{k+i}}{u_{2k+2l-i+1}}$ for $i=1,\ldots,l$ (see Figure~\ref{fig:kl-grid-free-appendix} with $k=2$ and $l=3$).
$\Gell$ has $n = 2(k+l)\ell + 2(k+l) - 19$ vertices. Note that $G'$ is a subgraph of $\Gell$.
If two edges of $G'$ cross each other, then the claim follows from Lemma~\ref{lemma:ellcompound}.

If no grid edge in $\Gell$ crosses the subgraph $G'$, then all grid edges must be drawn within the unique face of size $2k+2l$ in~$G'$ and hence the drawing is not $(k,l)$-grid-free. Therefore at least one grid edge $e = \edge{u_i}{u_{j}}$ must cross an edge $\edge{a}{b}$ of $G'$. More specifically, $e$ must cross an edge $\edge{a}{x}$, as $x$ connects to all vertices on the cycle $C$.

The edge $e$ connects two vertices $u_i$ and $u_{j}$ of $C$. There exist two paths in $C$ which connect $u_i$ and $u_{j}$. One of these path does not contain $a$. Now consider the closed curve $\gamma$ formed by this path and $e$. This curve partitions the plane into two or more regions. Since the edge $e$ crosses the edge $\edge{a}{x}$, $a$ and $x$ lie in different regions of $\gamma$. Hence the $\ell$-compound edge between $a$ and $x$ in $\Gell$ must cross $\gamma$ resulting in $\ell$ crossings.
\end{proof}
The maximum number of of edges in $(k,l)$-grid-free graphs with $n$ vertices is $c_{k,l}n$, where $c_{k,l}$ depends only on $k$ and $l$~\cite{maxedgenr_grid}.
Thus, for simple drawings, $\varrho_\klgridsub{k}{l} \le g(k,l)\cdot n^2$ for a computable function~$g$ depending only on $k$ and $l$.

\subsection{\boldmath $k$-Gap-Planar Graphs}
\label{sec:gap-planar}

In a $k$-gap-planar drawing every crossing is assigned to one of the two edges involved. No more than $k$ crossings can be assigned to a single edge. We first describe our lower bound construction for $k=1$ and then show how to extend it to larger $k$. As before we create a cycle $C$ of length $8$ of $\ell$-compound edges and connect each vertex of the cycle to an additional vertex $x$ using a further $8$ $\ell$-compound edges. Again by Lemma~\ref{lemma:ellcompound} none of the $\ell$-compound edges can cross. We add four single edges which connect diametrically opposite vertices of~$C$ (see Figure~\ref{fig:1gap-appendix}, left; $\ell$-compound edges are indicated in blue, single edges in black). In a $1$-gap-planar drawing one single edge has to be drawn in the same face (with respect to the cycle) as the vertex $x$ and must hence cross an $\ell$-compound edge, incurring $\ell$ crossings.

\begin{figure}[b]
    \centering
    \includegraphics[page=4]{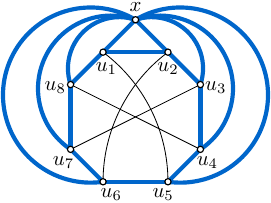}
    \hskip 2cm
    \includegraphics[page=5]{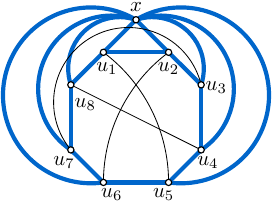}
    \caption{Left: drawing of $\Gell$ with $\crnum(G_\ell) \leq 6$; Right: 1-gap-planar drawing of $\Gell$ with $\crnum_\kgapsub{1}(\Gell) \geq  \ell$, $\ell$-compound edges are drawn with thick blue lines.}
    \label{fig:1gap-appendix}
\end{figure}

\begin{theorem}\label{theorem:1gap}
For every $\ell \geq 2$ there exists a 1-gap-planar graph $G_\ell$ with $n = 16\ell - 7$ vertices such that $\crnum_\kgapsub{1}(G_\ell) \geq \ell$ and $\crnum(G_\ell) \leq 6$. Thus, $\varrho_\kgapsub{1} \in \Theta(n).$
\end{theorem}
\begin{proof}
The upper bound follows directly from the fact that a $1$-gap-planar graph with $n$ vertices can have at most $5n-10$ edges~\cite{maxedgenr_gapplanar}.
Since each edge can have at most one crossing assigned to it, the number of crossings cannot exceed the number of edges. Thus, a $1$-gap-planar graph can have no more than $5n-10$ crossings. Clearly, if $\crnum_\kgapsub{1}(G) > \crnum(G)$ it must hold that $\crnum(G) \geq 2$ and therefore $\varrho_\kgapsub{1} \in O(n)$.

For the lower bound we first construct a graph $G'$ which consists of a cycle $C$ of length eight, $C = \langle u_1,u_2, \dots ,u_8 \rangle$ and a vertex $x$ connected to each vertex of $C$. Now we construct a graph $\Gell$ by replacing each edge in $G'$ by an $\ell$-compound edge. Finally, we add the \emph{single} edges $\edge{u_i}{u_{i+4}}, i=1,2,3,4$ (see Figure~\ref{fig:1gap-appendix}). $\Gell$ has $n = 16(\ell-1) + 9 = 16\ell - 7$ vertices. Note that $G'$ is a subgraph of $\Gell$. If two edges of $G'$ cross each other, then the claim follows from Lemma~\ref{lemma:ellcompound}.

If no single edge crosses the subgraph~$G'$, then all single edges must be drawn within the unique face of size eight in~$G'$. All single edges cross each other, causing six crossings, which need to be mapped to the four single edges. This is not possible and hence the drawing is not 1-gap-planar.
Therefore, at least one single edge $e = \edge{u_i}{u_{j}}$ must cross an edge $\edge{a}{b}$ of~$G'$. More specifically, $e$ must cross an edge $\edge{a}{x}$, as $x$ connects to all vertices on the cycle $C$.

The edge $e$ connects two vertices $u_i$ and $u_{j}$ of $C$. There exist two paths in $C$ which connect $u_i$ and $u_{j}$. One of these path does not contain $a$. Now consider the closed curve $\gamma$ formed by this path and $e$. This curve partitions the plane into two or more regions. Since the edge $e$ crosses the edge $\edge{a}{x}$, $a$ and $x$ lie in different regions of $\gamma$. Hence the $\ell$-compound edge between $a$ and $x$ in $\Gell$ must cross $\gamma$ resulting in $\ell$ crossings.
\end{proof}

In general, a $k$-gap-planar graph with $n$ vertices can have at most $O(\sqrt{k}n)$ edges~\cite{maxedgenr_gapplanar}.
Since each edge can have at most $k$ crossings assigned to it, the number of crossings cannot exceed the number of edges times $k$. Thus, a $k$-gap-planar graph can have no more than $O(k\sqrt{k}n)$ crossings and therefore $\varrho_\kgapsub{k} \in O(k\sqrt{k}n)$.
To prove a lower bound for $k$-gap-planarity, we can use exactly the same construction as above starting with a cycle of length~$8k$. Our lower bound is cubic in $1/k$ and hence does not match the upper bound.

\begin{corollary}\label{theorem:kgap}
For every $\ell \geq 2$ there exists a $k$-gap-planar graph $G_\ell^k$ with $n = 16k\ell - 8k+1$ vertices such that $\crnum_\kgapsub{k}(G_\ell^k) \geq \ell$ and $\crnum(G_\ell^k) \leq 8k^2 - 2k$. Thus, $\varrho_\kgapsub{k} \in \Omega(n/k^3).$
\end{corollary}
\begin{proof}
We first construct the graph $G^k$ starting with a cycle $C = \langle u_1, u_2,\ldots, u_{8k} \rangle$ of length $8k$ and a vertex $x$ connected to each vertex of $C$. Now we construct a graph $\Gell^k$ by replacing each edge in $G^k$ by an $\ell$-compound edges. Finally, we add the single edges $\edge{u_i}{u_{i+4k}}, 1\leq i \leq 4k$.

If all $4k$ single edges are drawn inside the cycle $C$ they all mutually cross, resulting in
${4k\choose 2} = \frac{1}{2} (4k) (4k - 1) = 8k^2 - 2k$ crossings. Since $8k^2 - 2k > 4k^2$
for $k > \frac{1}{2}$ the crossings cannot be mapped to the $4k$ edges without exceeding a load of $k$ crossings per edge. Hence at least one single edge must cross an edge of $G^k$ and the argument follows as above.
\end{proof}

\section{\boldmath Skewness-$k$ and $k$-Apex Graphs}
\label{sec:skew-apex}

In this section we further explore the use of $\ell$-compound edges for the construction of lower bound examples. In particular, we consider the two families of \mbox{skewness-$k$} graphs and $k$-apex graphs.
As in the previous section, we exploit the fact that crossings of $\ell$-compound edges result in $\ell$ crossings, which allows us to enforce a structure on crossing-minimal drawings.

\subsection{\boldmath Skewness-$k$ Graphs}
\label{sec:skewness-k}

In a skewness-$k$ drawing all crossings must be \emph{covered} by a specific set of most $k$ edges.
A crossing is covered by an edge when the edge is part of the crossing.
In particular, in a skewness-$1$ drawing, there is one \emph{special} edge that covers all crossings. Skewness-$1$ drawings are also known as almost planar or near planar drawings.
We first describe our lower bound construction for $k= 1$ and then show how to extend it to larger $k$.

\begin{theorem}\label{theorem:skew}
For every $\ell \geq 3$ there exists a skewness-1 graph $\Gell$ with $n = 11\ell - 4$ vertices such that $\crnum_\kskewsub{1}(\Gell) \geq \ell$ and $\crnum(\Gell) \leq 3$. Thus, $\varrho_\kskewsub{1} \in \Theta(n)$.
\end{theorem}

\begin{figure}[t]
    \centering
    \includegraphics[page=4]{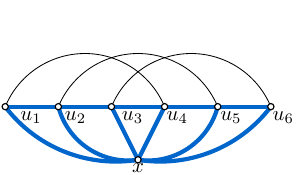}
    \hskip 2cm
    \includegraphics[page=5]{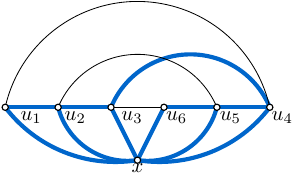}
    \caption{Left: drawing of $\Gell$ with $\crnum(\Gell) \leq 3$; Right: skewness-1 drawing of $\Gell$ with $\crnum_\kskewsub{1}(\Gell) \geq \ell$, $\ell$-compound edges are drawn with thick blue lines.}
    \label{fig:1-skewness-appendix}
\end{figure}

\begin{proof}
The upper bound directly follows from the fact that a skewness-$1$ graph with $n$ vertices can have at most $3n-5$ edges~\cite{Didimo_2019}. Since one edge must cover all crossings, there can not be more crossings than the number of edges.
Thus, a crossing-minimal skewness-$1$ drawing  cannot have more than $3n - 4$ crossings and therefore $\varrho_\kskewsub{1} \in O(n)$.

For the lower bound we first construct a graph $G'$ which consists of
a path $P =  \langle u_1 \ldots u_6 \rangle$ of length six and a vertex $x$ connected to each vertex of $P$. We construct a graph $\Gell$ by replacing each edge in $G'$ by an $\ell$-compound edge. Finally, we add the single edges $\edge{u_1}{u_4}$, $\edge{u_2}{u_5}$ and $\edge{u_3}{u_6}$ (see Figure~\ref{fig:1-skewness-appendix}). $\Gell$ has $n = 11\ell - 11 + 7 = 11\ell - 4$ vertices. Note that $G'$ is a subgraph of $\Gell$. If two edges of $G'$ cross each other, then the claim follows from Lemma~\ref{lemma:ellcompound}.

Crossing-minimal skewness-1 drawings must be simple.
Assume for contradiction that a crossing-minimal skewness-1 drawing is not simple. Then there must be two edges $e$ and $e'$ which have more than one point in common. One of these edges, say $e$, must be the special edge. Hence $e'$ cannot cross any other edges than $e$. We can reroute $e$ along $e'$ between two of their share points, thereby reducing the number of crossings.

Assume that each edge $\edge{u_i}{u_{i+1}}$ of $P$ defines a 3-cycle with $x$ and its incident edges such that all other vertices lie in the same face bounded by the 3-cycle.
We argue that in this case there is no possible simple skewness-$1$ drawing.
We start with the crossing-free drawing $\Gamma$ of the edges connecting the endpoints of $\ell$-compound edges
and we add the edge $\edge{u_1}{u_4}$.
Since we can assume that the drawing produced is simple, $\edge{u_1}{u_4}$ can only be drawn with 0 crossing or with at least two crossings. If it is drawn with at least two crossings, then $\edge{u_1}{u_4}$ is the special edge. Symmetrically, the same holds for $\edge{u_3}{u_6}$. Since at least one of them is not the special edge, at least one of them is drawn with 0 crossings with respect to the $\ell$-compound edges, without loss of generality we assume that $\edge{u_1}{u_4}$ is drawn without crossings. Now, there is a 3-cycle $(u_1,u_4,x)$ that divides both vertex pairs $\edge{u_2}{u_5}$ and $\edge{u_3}{u_6}$. Both edges will have to cross the 3-cycle, thus one of the edges of the 3-cycle has to be the special edge.
In particular, this implies that neither  $\edge{u_2}{u_5}$ nor $\edge{u_3}{u_6}$ are the special edge.

By the argument above, $\edge{u_3}{u_6}$ cannot cross $\Gamma$. Therefore, $\edge{u_3}{u_6}$, has to cross $\edge{u_1}{u_4}$, making $\edge{u_1}{u_4}$ the special edge.
Now $u_2$ and $u_5$ are separated by the 3-cycle $(u_3,u_6,x)$. Since the special edge is not part of this cycle, it is not possible for $\Gamma$ to be a skewness-1 drawing.
An illustration can be found in Figure~\ref{fig:1-skewness-appendix} on the left.

If there is a skewness-1 drawing where all crossings are between single edges, then all vertices must lie in the same face of $G'$. By the above argument, such a drawing is not skewness-1.
By Lemma~\ref{lemma:ellcompound} if two edges of $G'$ cross, the drawing has $\ell$ crossings.
Thus, in a crossing-minimal skewness-1 drawing $\Gamma'$, at least one single edge $e$ has to cross an edge~$\edge{a}{b}$ of~$G'$. One of these edges is the special edge.
Since $\Gamma'$ is simple, $a$ and $b$ must lie in different faces of $\Gamma'$ and the $\ell-1$ paths connecting $a$ and $b$ must cross $e$, resulting in at least $\ell$ crossings.
\end{proof}
We extend the proof for the skewness-1 crossing ratio to the skewness-$k$ crossing ratio.

\begin{theorem}
\label{theorem:kskew}
For every $\ell \geq 3$ there exists a skewness-$k$ graph $\Gell^k$ with $n = 11\ell k - 5k + 1$ vertices such that $\crnum_\kskewsub{k}(\Gell^k) \geq \ell k$ and $\crnum(\Gell^k) \leq 3k$. Thus, $\varrho_\kskewsub{k} \in \Omega(n/k)$.
\end{theorem}

\begin{proof}
We construct a graph $\Gell^k$ by identifying the vertex $x$ of $k$ copies of $\Gell$ (see Figure~\ref{fig:k-skewness-appendix}).
A drawing of $\Gell^k$ is skewness-$k$ if and only if all $k$ instances of $\Gell$ are skewness-1.
Therefore, applying the above proof on all instances gives $\crnum_\kskewsub{k}(\Gell^k) = k\ell$. However, if skewness-$k$ is not enforced, there exists a drawing with $\crnum(\Gell^k) \leq 3k$ where each $\Gell$ is drawn with three crossings (see Figure~\ref{fig:k-skewness-appendix}, left).
\end{proof}
A skewness-$k$ graph with $n$ vertices can have at most $3n-6 + k$ edges~\cite{Didimo_2019}.
Since at most $k$ edges are involved in all crossings, there can not be more crossings than the number of edges times~$k$.
Thus, a crossing-minimal skewness-$k$ drawing  cannot have more than $3kn - 6k + k^2$ crossings and therefore $\varrho_\kskewsub{k} \in O(kn+k^2)$.

\subsection{\boldmath $k$-Apex Graphs}
\label{sec:k-apex}
In a $k$-apex drawing all crossings must be covered by at most $k$ vertices. A crossing is covered by a vertex when an edge incident to that vertex is part of the crossing.
Every skewness-$k$ drawing is also $k$-apex. Hence we know that there is a graph~$\Gell$ with $n = 11\ell - 4$ vertices which admits a drawing with three crossings (see Figure~\ref{fig:1-skewness-appendix}, left) and admits a 1-apex drawing with $\ell$ crossings (see Figure~\ref{fig:1-skewness-appendix}, right).

Moreover, the proof of Theorem~\ref{theorem:skew} of the lower bound on the skewness-1 crossing ratio translates to 1-apex drawings.
In the case in which each edge $\edge{u_i}{u_{i+1}}$ of $P$ defines a 3-cycle with $x$ and its incident edges such that all other vertices lie in the same face bounded by the 3-cycle, the drawing cannot be 1-apex, as removing any vertex could still not produce a plane drawing (see Figure~\ref{fig:1-skewness-appendix}, left).
The rest of the proof directly translates to this setting.

\begin{figure}[t]
    \centering
    \includegraphics[page=6]{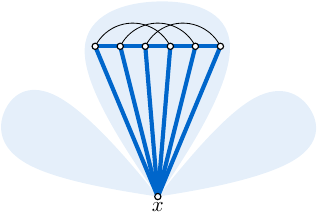}
    \hskip 2cm
    \includegraphics[page=7]{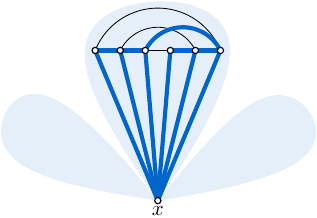}
    \caption{Left: drawing of $\Gell$ with $\crnum(\Gell) \leq 3k$; Right: skewness-$k$ drawing of $\Gell$ with $\crnum_\kskewsub{k}(\Gell) \geq k\ell$,
    each of the $k$ petals contains a copy of $\Gell$, resulting in a graph $\Gell^k$, $\ell$-compound edges are drawn with thick blue lines.}
    \label{fig:k-skewness-appendix}
\end{figure}

\begin{corollary}
\label{corollary:1apex}
For every $\ell \geq 3$ there exists a 1-apex graph $\Gell$ with $n = 11\ell - 4$ vertices such that $\crnum_\kapexsub{1}(\Gell) \geq \ell$ and $\crnum(\Gell) \leq 3$. Thus, $\varrho_\kapexsub{1} \in \Omega(n)$.
\end{corollary}

This bound is not tight, as the upper bound is quadratic. A $1$-apex graph with $n$ vertices can have at most $4n-10$ edges~\cite{Didimo_2019}. Thus, a $1$-apex crossing-minimal drawing has at most $(4n-10)^2$ crossings and therefore $\varrho_\kapexsub{1} \in O(n^2)$.

Similar to the 1-apex example, we know that for each $k$ there is a graph $\Gell^k$ with $n = 11\ell k - 5k + 1$ vertices which admits a drawing with $3k$ crossings and admits a $k$-apex drawing with $k\ell$ crossings.
As before, the proof for skewness-$k$ translates to $k$-apex, because the drawing is only $k$-apex when every copy of $\Gell$ is drawn 1-apex.

\begin{corollary}
\label{corollary:kapex}
For every $\ell \geq 3$ there exists a $k$-apex graph $\Gell^k$ with $n = 11\ell k - 5k + 1$ vertices such that $\crnum_\kapexsub{k}(\Gell^k) \geq \ell k$ and $\crnum(\Gell^k) \leq 3k$. Thus, $\varrho_\kapexsub{k} \in \Omega(n/k)$.
\end{corollary}
A $k$-apex graph with $n$ vertices can have at most $3(n-k)-6 + \sum^k_{i=1}(n-i) = 3(n-k)-6 + kn - \frac{k(k+1)}{2}$ edges~\cite{Didimo_2019}. Thus, a $k$-apex crossing-minimal drawing has at most $\left(3(n-k)-6 + kn - \frac{k(k+1)}{2}\right)^2$ crossings and therefore $\varrho_\kapexsub{k} \in O(k^2n^2 - k^4)$.

\section{\boldmath Planarly Connected Graphs, $k$-Fan-Crossing-Free Graphs and Straight-line RAC-Graphs}
\label{sec:plan-con-fan-free}

In this section we introduce the concept of \emph{$\ell_2$-bundles} and \emph{$\ell_3$-bundles}, which are sets of $\ell$ paths of length two or three, respectively, all connecting the same two vertices (see Figure~\ref{fig:ellx-bundle}).
We indicate $\ell_2$-bundles and $\ell_3$-bundles with thick red and green lines, respectively.
The vertices that are interior to a bundle have degree 
two.
Two crossing bundles result in
a quadratic number of
crossings. We show how to enforce such crossings to prove tight bounds on the crossing ratio. Lemma~\ref{lemma:bundle_plcon} argues that there is a planarly connected drawing that is crossing minimal in which $\ell_2$-bundles
do not cross themselves and
no vertex lies between two consecutive paths of the same $\ell_2$-bundle.

\begin{figure}[h]
    \centering
    \begin{subfigure}{.45\textwidth}
        \centering
        \includegraphics[page=4]{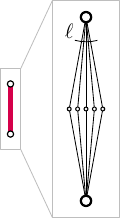}
    \end{subfigure}
    \begin{subfigure}{.45\textwidth}
        \centering
        \includegraphics[page=5]{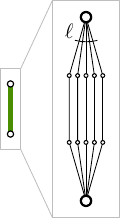}
    \end{subfigure}
    \caption{Left: an $\ell_2$-bundle; Right: an $\ell_3$-bundle; drawn as a thick red or green line, respecitvely.}
    \label{fig:ellx-bundle}
\end{figure}

\subsection{Planarly Connected Graphs}
\label{sec:plan-con}

In a planarly connected drawing two edges can cross only if their endpoints are connected by another edge which has no crossings.
We show that $\varrho_\plconsub \in \Theta(n^2)$
 by constructing a graph in which a crossing occurs between two edges whose endpoints are not connected when two vertices are drawn in the same face of a cycle composed of $\ell_2$-bundles. When drawing the vertices in separate faces, two $\ell_2$-bundles have to cross, resulting in a quadratic number of crossings.
\begin{restatable}{lemma}{plconbundle}
\label{lemma:bundle_plcon}
If a planarly connected graph $G$ contains $\ell_2$-bundles, then there is a planarly connected crossing-minimal drawing $\Gamma$ of $G$ in which paths from the same $\ell_2$-bundle do not cross each other, all paths of the same $\ell_2$-bundle cross the same edges, and
no vertex lies between two consecutive paths of the same $\ell_2$-bundle.
\end{restatable}

\begin{proof}
We first prove that in every planarly connected crossing-minimal drawing of $G$, paths from the same $\ell_2$-bundle do not cross each other.

Suppose there is a drawing where two paths $p_1$ and $p_2$ of an $\ell_2$-bundle cross. One of the paths, say $p_1$, has less or equal crossings than the other path, say $p_2$. If we draw $p_2$ parallel and sufficiently close to $p_1$, such that there are no vertices between $p_1$ and $p_2$ and they both cross the same edges, then the number of crossings does not increase. However, since $p_1$ and $p_2$ do not cross anymore, the number of crossings decreases by at least one.
Recall that, by definition, in a graph containing $\ell_2$-bundles the vertices interior to a bundle  have degree exactly two.
Thus, in a planarly connected drawing it is always possible to redraw a path $p$ in an $\ell_2$-bundle parallel and sufficiently close to another path $p'$ in that bundle.
Therefore, in a planarly connected crossing minimal drawing edges of the same $\ell_2$-bundle do not cross.

The redrawing argument can also be used to
prove that there is a planarly connected crossing-minimal drawing $\Gamma$ of $G$ in which
all paths from the same $\ell_2$-bundle cross the same edges and no vertex is drawn in between two consecutive paths of the same $\ell_2$-bundle.
Given a planarly connected crossing-minimal drawing of $G$,
for each $\ell_2$-bundle
we can select a path with the minimum number of crossings and redraw all the other paths close to it
such that they all cross the same edges and no vertex lies in the small gap between one path and the next.
\end{proof}
\begin{figure}[b]
    \centering
    \begin{subfigure}{.45\textwidth}
        \centering
        \includegraphics[page=7]{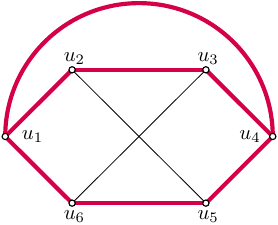}
        \label{fig:planarly-connected-appendix-left}
    \end{subfigure}
    \begin{subfigure}{.45\textwidth}
        \centering
        \includegraphics[page=8]{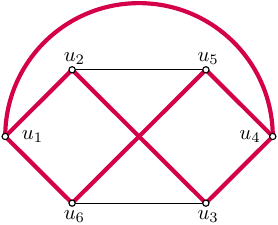}
        \label{fig:planarly-connected-appendix-right}
    \end{subfigure}
    \caption{Left: drawing of $\Gell$ with $\crnum(G_\ell) = 1$; Right: planarly connected drawing of $\Gell$ with $\crnum_\plconsub(\Gell) = \ell^2$, $\ell_2$-bundles are drawn with thick red lines.}
    \label{fig:planarly-connected-appendix}
\end{figure}
\begin{theorem}\label{theorem:planarcon}
For every $\ell \geq 2$ there exists a planarly connected graph $G_\ell$ with $n = 7\ell + 6$ vertices such that $\crnum_\plconsub(G_\ell) \geq \ell^2$ and $\crnum(G_\ell) \leq 1$. Thus, $\varrho_\plconsub \in \Theta(n^2).$
\end{theorem}

\begin{proof}
The upper bound $\varrho_\plconsub \in O(n^2)$ follows directly from the fact that a planarly connected graph with $n$ vertices has at most $cn$ edges, where $c$ is a  constant~\cite{maxedgenr_pl-con}.

For the lower bound, we construct a graph $\Gell$.
We start with an 6-cycle $\langle u_1,u_2, \dots ,u_6 \rangle$, and an edge $\edge{u_1}{u_4}$, resulting in the graph $G'$. Then we replace each edge in $G'$ with an $\ell_2$-bundle.
Finally, add regular edges $\edge{u_1}{u_4}$, $\edge{u_2}{u_5}$ and $\edge{u_3}{u_6}$.
The graph $\Gell$ has $n= 7\ell + 6$ vertices and admits a drawing with only one crossing (see Figure~\ref{fig:planarly-connected-appendix}, left).

By Lemma~\ref{lemma:bundle_plcon} there is a planarly connected crossing-minimal drawing $\Gamma$ of $\Gell$
in which for every $\ell_2$-bundle
its edges do not cross each other,
all paths cross the same edges, and no vertex lies between two consecutive paths.
Let $C$ be the cycle of $\ell_2$-bundles $C = \langle u_1,u_4,u_5,u_6 \rangle$.
If $C$ crosses itself in $\Gamma$, then there are $\ell^2$ crossings.
Otherwise $C$ is drawn without crossings in $\Gamma$ and defines two faces of size eight.
The other faces do not contain vertices of $\Gamma$ in their interior.
Thus,~$u_2$ and~$u_3$ can either be drawn in the same face defined by $C$, or in different faces.
If they are drawn in the same face, $\Gamma$ cannot be planarly connected: a simple case distinction shows that there has to be a crossing between two edges whose endpoints are not connected by any other edge.
Specifically, there is a crossing between either $\edge{u_1}{u_2}$ and $\edge{u_3}{u_4}$, $\edge{u_1}{u_2}$ and $\edge{u_3}{u_6}$, $\edge{u_2}{u_5}$ and $\edge{u_3}{u_6}$ or $\edge{u_2}{u_5}$ and $\edge{u_3}{u_4}$. These are all not allowed, as each time there is no edge connecting the two edges.
If $u_2$ and~$u_3$ are drawn in different faces defined by $C$,
then the $\ell_2$-bundle $\edge{u_2}{u_3}$ has to cross $C$. As~$C$ is composed of $\ell_2$-bundles,
there are $\ell$ cycles, distinct in edges, all separating $u_3$ and $u_4$. All the $\ell$ distinct paths in the $\ell_2$-bundle $\edge{u_2}{u_3}$ have to cross the $\ell$ distinct (in edges) cycles in~$C$. Therefore,
$\Gamma$ has at least $\ell^2$ crossings.
\end{proof}

\subsection{\boldmath $k$-Fan-Crossing-Free Graphs}\label{sec:fan-free}
\begin{wrapfigure}[9]{r}{0.1\linewidth}
  \vspace{-.4cm}
  \raggedleft
  \includegraphics[page=7]{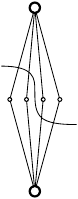}
\end{wrapfigure}
In a $k$-fan-crossing-free drawing it is forbidden for an edge to cross $k$ edges which are adjacent to the same vertex.
While it is possible for an edge to ``weave'' through an $\ell_2$-bundle with $\ell \leq 2k-2$ (see figure on the right), no edge can cross an $\ell_2$-bundle with $\ell \geq 2k-1$.
In contrast, $\ell_3$ bundles can be crossed for any $\ell$.
Lemma~\ref{lemma:bundle_fanfree} argues that, given a sufficiently large $\ell$, there is a fan-crossing-free drawing that is crossing minimal in which $\ell_2$-bundles
do not cross themselves and
no vertex lies between two consecutive paths of the same $\ell_2$-bundle.
We show that $\varrho_\fanfreesub \in \Omega(n^2)$ and $\varrho_\kfanfreesub{k} \in \Omega(n^2/k^3)$ using $\ell_2$-bundles to force $\ell_3$-bundles to cross.

\begin{restatable}{lemma}{fanfreebundle}
\label{lemma:bundle_fanfree}
Let $G$ be a $k$-fan-crossing-free graph and let $p$ be the length of the longest simple path in $G$.
If $G$ contains $\ell_2$-bundles with $\ell \geq 
2p(k-1) +1$, then in any crossing-minimal $k$-fan-crossing-free drawing of $G$
the edges in the $\ell_2$-bundles have no crossings.
\end{restatable}

\begin{proof}

\begin{figure}[b]
    \centering
    \begin{subfigure}{.48\textwidth}
        \centering
        \includegraphics[page=1]{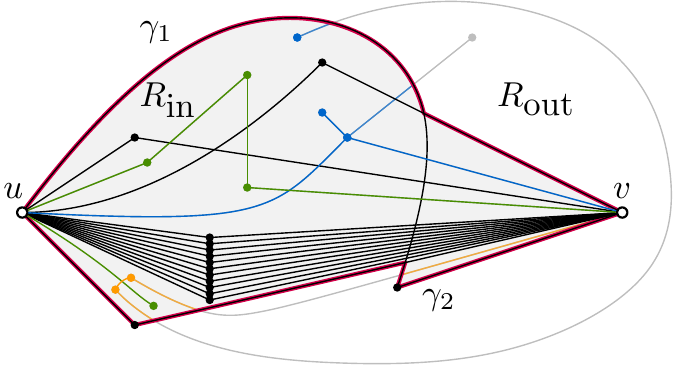}
        \label{fig:l2bundles-kfanfree-left}
    \end{subfigure}\hfill
    \begin{subfigure}{.48\textwidth}
        \centering
        \includegraphics[page=2]{figures/lemma_l2bundles.pdf}
        \label{fig:l2bundles-kfanfree-right}
    \end{subfigure}
    \caption{
    Left: An $\ell_2$-bundle $\mathcal{B}$ crossed by itself and by other edges.
    The four parts in which $\Gamma_R$ is partitioned are drawn in different colors:
    black, green, blue, and yellow.
    Right: Redrawing the different (parts of) edges in $\Gamma_R$. This removes the crossings of $\mathcal{B}$.
    }
    \label{fig:l2bundles-kfanfree}
\end{figure}

Let $\Gamma$ be a crossing-minimal $k$-fan-crossing-free drawing of $G$ and let $\Lambda$ be the drawing of an $\ell_2$-bundle $B$ in~$\Gamma$.
We assume for the sake of contradiction that $\Lambda$ has crossings in~$\Gamma$. These crossings can involve solely edges of $B$ and also other edges of $G$.
Let $u$ and $v$ denote the two endpoints of the bundle $B$
The drawing $\Lambda$ is bounded by two curves $\gamma_1$ and $\gamma_2$ which connect $u$ and $v$. These curves are formed by the drawings of (parts of) edges of $B$ (see Figure~\ref{fig:l2bundles-kfanfree}) and they bound a region $R$ of the plane, which is exactly the region in which $B$ is drawn. Our goal is to redraw $G$ without increasing the number of crossings such that the drawing remains $k$-fan-crossing-free and $(i)$ all parts of $G$ which are not $B$ are drawn in a $\varepsilon$-strip along the boundary of $R$
and $(ii)$ the bundle $B$ is drawn crossing-free in the interior of $R$. Once the former is achieved the latter is straightforward.

Consider the \emph{restriction} $\Gamma_R$ of the drawing $\Gamma$ to $R$, that is, only those parts of $\Gamma$ that lie strictly in the interior of $R$. Specifically, neither the vertices $u$ and $v$ nor the curves $\gamma_1$ and $\gamma_2$ are part of $\Gamma_R$. We say that two vertices of $G$ are $R$-connected if they are connected by a path which is drawn fully in the interior of $R$. We can partition $\Gamma_R$
into four parts: edges and vertices of $B$, $R$-connected components that lie full inside $R$, $R$-connected components that touch the interior of $\gamma_1$, and $R$-connected components that touch the interior of $\gamma_2$. For example, the green path in Figure~\ref{fig:l2bundles-kfanfree} lies fully in the interior of $R$, the blue $R$-connected components touch the interior of $\gamma_1$, and the yellow $R$-connected components touch the interior of $\gamma_2$.

The drawing $\Gamma$ is $k$-fan-crossing-free and $\ell > 2p(k-1)$, that is, $\ell$-bundles are more than $2(k-1)$ wider than the longest path in $G$. Any $R$-connected component that touches both $\gamma_1$ and $\gamma_2$ must contain parts of the drawing of a path in $G$ that cross all paths in $B$.
Since in a $k$-fan-crossing-free drawing an edge can cross at most $k-1$ edges incident to a vertex, in particular $u$ or $v$,
in total the edges of a path can cross at most $2p(k-1)$ edges of $B$.
Hence no $R$-connected component can touch both $\gamma_1$ and $\gamma_2$.

We can independently and locally redraw the edges of $G$ that correspond to each part, $B$ in the interior of $R$, and the other three parts within a $\varepsilon$-strip along the boundary of $R$. Note that independent fragments of a single edge of $G$ can appear in two parts and will be redrawn independently.
\end{proof}

\begin{figure}[b]
    \centering
    \begin{subfigure}{.45\textwidth}
        \centering
        \includegraphics[page=4]{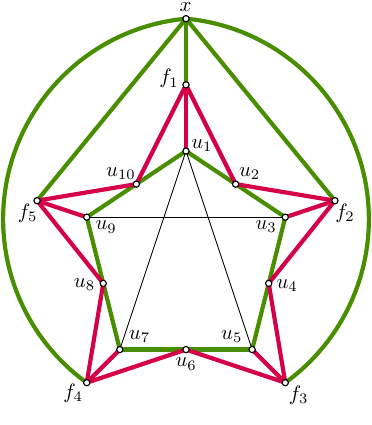}
        \label{fig:fan-crossing-free-appendix-left}
    \end{subfigure}\hfill
    \begin{subfigure}{.45\textwidth}
        \centering
        \includegraphics[page=6]{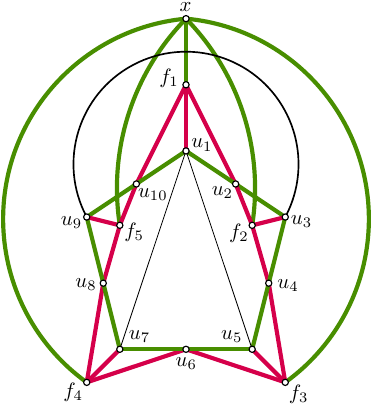}
        \label{fig:fan-crossing-free-appendix-right}
    \end{subfigure}
    \caption{
    Left: drawing of $\Gell$ with $\crnum(G_\ell) \leq 2$; Right: fan-crossing-free drawing of $\Gell$ with $\crnum_\fanfreesub(\Gell) \geq \ell^2$.
    Thick red lines represent
    $\ell_2$-bundles and thick green lines $\ell_3$-bundles.}
    \label{fig:fan-crossing-free-appendix}
\end{figure}

\begin{theorem}\label{theorem:2fan}
For every $\ell \geq 95$ there exists a fan-crossing-free graph $G_\ell$ with $n = 45\ell + 16$ vertices such that $\crnum_\fanfreesub(G_\ell) \geq \ell^2/2$ and $\crnum(G_\ell) \leq 2$. Thus, $\varrho_\fanfreesub \in \Theta(n^2)$.
\end{theorem}

\begin{proof}
The upper bound $\varrho_\plconsub \in O(n^2)$ follows directly from the fact that a fan-crossing-free  graph  with $n$ vertices  can  have  at  most  $4n-8$  edges~\cite{maxedgenr_fanfree}.

For the lower bound we construct a graph $\Gell$ (see Figure~\ref{fig:fan-crossing-free-appendix}).
We start with a 10-cycle $C = \langle u_1,u_2 \dots, u_{10} \rangle$. We add five vertices $f_1,f_2,\dots, f_5$ and add the edges $\edge{f_i}{u_{2i-2}}, \edge{f_i}{u_{2i-1}}, \edge{f_i}{u_{2i}}$ for $i = 1,2,\dots,5$. We add a new vertex $x$ and connect it to each $f_i$ for $i = 1,2,\dots,5$.
We then replace each edge $\edge{f_i}{u_j}$
by an $\ell_2$-bundle and replace each edge $\edge{x}{f_i}$ and $\edge{u_i}{u_j}$ 
by an $\ell_3$-bundle. Finally, we add the  regular edges $\edge{u_1}{u_{5}}, \edge{u_1}{u_{7}}$, and $\edge{u_3}{u_{9}}$.
%
The resulting graph $\Gell$ has $n= 45\ell + 16$ vertices and admits a drawing with only two crossings (see Figure~\ref{fig:fan-crossing-free-appendix}, left).

In $\Gell$ there is a cycle $F = \langle f_1,u_2,f_2,u_4,f_3,u_6,f_4,u_8,f_5,u_{10} \rangle$
of $\ell_2$-bundles.
There are $10 + 5 + 1 = 16$ vertices which are not interior to bundles. The longest path can use at most two vertices inside each bundle and has length at most $48-1 = 47$.
Since $\ell \geq 95 = 2 \cdot 47 \cdot(2-1) + 1$, by Lemma~\ref{lemma:bundle_fanfree}, all $\ell_2$-bundles must be drawn without crossings and $F$ must be drawn plane.
If $u_1$, $u_3$, $u_5$, $u_7$ and $u_9$ all lie in the same face of $F$ then we are effectively in the situation depicted in Figure~\ref{fig:fan-crossing-free-appendix} (left) and the drawing is not fan-crossing-free. Hence $u_1$, $u_3$, $u_5$, $u_7$, and $u_9$ must distribute over the two faces of $F$. Specifically, we can assume that there a vertex $u_i$, $i \in \{ 1, 3, 5, 7, 9\}$, such that $u_i$ lies in the same face of $F$ as $x$.

We consider the cycle $Z$ consisting of the two $\ell_2$-bundles between $u_i$ and $f_{(i+1)/2}$ and between $f_{(i+1)/2}$ and $u_{i-1}$, and the $\ell_3$-bundle between $u_i$ and $u_{i-1}$. Symmetrically, we can consider the cycle $Z'$ involving $u_{i+1}$ instead of $u_{i-1}$. The vertex $x$ can lie within any face of the $\ell_3$-bundle. Consider the cycles formed by each path of the $\ell_3$-bundle together with the two $\ell_2$-bundles. These cycles divide the plan into two faces. We say that the \emph{exterior} face of each cycle is the face that contains $F$.
We say that $x$ lies in the \emph{exterior} of $Z$ ($Z'$, respectively), if it lies in the exterior face of $\ell/2$ of these cycles.
Otherwise, $x$ lies in the \emph{interior} of $Z$ ($Z'$, respectively). Observe that $x$ lies either in the interior of $Z$, or in the interior of $Z'$, or in the exterior of both.

If $x$ lies in the exterior face of both $Z$ and $Z'$, then the $\ell_3$-bundle connecting $x$ and $f_{(i+1)/2}$ crosses at least $\ell/2$ paths in either $Z$ or $Z'$ (see Figure~\ref{fig:fan-crossing-free-appendix}, right; $u_i = u_3$ and $f_{(i+1)/2} = f_2$).
If $x$ lies in the interior face of $Z$, then (by definition of the faces of $Z$) there is an $f_j \neq f_{(i+1)/2}$ which lies in the exterior face of $Z$. The $\ell_3$-bundle connecting $x$ and $f_j$ crosses at least $\ell/2$ paths in $Z$. The argument for $Z'$ is symmetric.
\end{proof}
The proof can be extended to $k$-fan-crossing-free graphs by creating a cycle~$C$ of length $6+2k$ which results in a graph with $3/2 \cdot (6+2k)+1 = 10 + 3k$ vertices not interior to bundles and a longest path of length at most $29 + 9k$.
Since no path should be able to cross the $\ell_2$-bundles, we require that $\ell \geq 2 \cdot (29 + 9k) \cdot(k-1) + 1 = 18k^2 + 40k -57$.

\begin{corollary}\label{corollary:kfan}
For every $\ell \geq 18k^2 + 40k -57$ there exists a $k$-fan-crossing-free graph $G_\ell^k$ with $n = 9k\ell + 27\ell  + 3k + 10$ vertices such that $\crnum_\kfanfreesub{k}(G_\ell^k) \geq \ell^2$ and $\crnum(G_\ell^k) \leq k$.
Thus, $\varrho_\kfanfreesub{k} \in \Omega(n^2/k^3)$.
\end{corollary}

This bound is not tight in $k$. A $k$-fan-crossing-free graph with $n$ vertices can have at most $3(k-1)(n-2)$ edges~\cite{maxedgenr_fanfree}. Thus, a crossing-minimal $k$-fan-crossing-free drawing cannot have more than $3(k-1)(n-2)^2$ crossings. This yields $k_{\kfanfreesub{k}} \in O(k^2n^2)$.

\subsection{Straight-line RAC Graphs}
\label{sec:rac}
A straight-line RAC drawing (Right Angle Crossings)  is a straight-line drawing of a graph in which  any two crossing edges form a $\frac{\pi}{2}$ angle at their crossing point. In a straight-line RAC drawing, it is not possible for two adjacent edges to be crossed by the same edge. Thus, we observe that a straight-line RAC drawing is also fan-crossing-free.

The proof and construction provided for the bound on the fan-crossing-free crossing ratio, also holds for the straight-line RAC crossing ratio. Every time the proof states that the drawing cannot be fan-crossing-free, the drawing is also not straight-line RAC. Thus, we can conclude that a crossing minimal straight-line RAC drawing $\Gamma$ of $\Gell$ in Theorem~\ref{theorem:2fan}  contains at least $\ell^2$ crossings. In Figure~\ref{fig:rac} we show such a drawing $\Gamma$ with $\ell^2$ crossings exists. To show that the crossings do form angles of $\frac{\pi}{2}$, in Figure~\ref{fig:rac} we draw the $\ell_3$-bundles $\edge{x}{f_1}$, $\edge{x}{f_2}$, $\edge{x}{f_5}$, $\edge{u_2}{u_3}$ and $\edge{u_9}{u_{10}}$ in detail for $\ell=3$.

A straight-line RAC graphs can have at most $4n-10$ edges~\cite{racdrawings11}. Thus, straight-line drawings of RAC graphs have at most $(4n-10)^2$ crossings, giving $\varrho_\racsub \in O(n^2)$.

\begin{corollary}
\label{corollary:rac}
For every $\ell \geq 95$ there exists a straight-line RAC graph $G_\ell$ with $n = 45\ell + 16$ vertices such that $\crnum_\racsub(G_\ell) \geq \ell^2$ and $\crnum(G_\ell) \leq 2$. Thus, $\varrho_\racsub \in \Theta(n^2)$.
\end{corollary}

\begin{figure}[t]
    \centering
    \includegraphics[page=6]{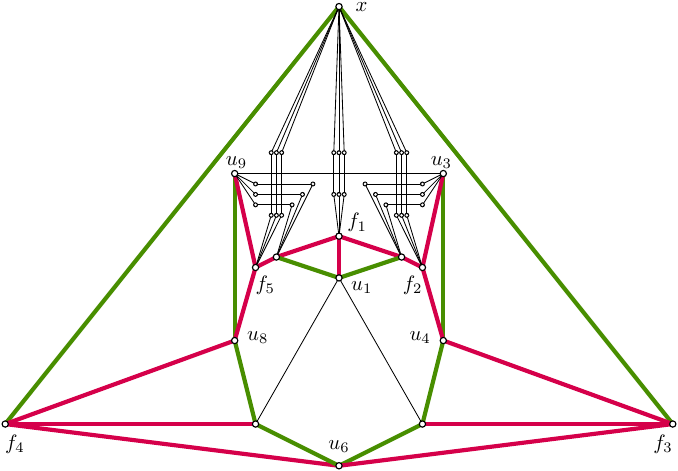}
    \caption{Straight-line RAC drawing of $\Gell$ with $\crnum_\racsub(\Gell) \geq \ell^2$ for $\ell = 3$. Thick red lines represent
    $\ell_2$-bundles and thick green lines $\ell_3$-bundles.}
    \label{fig:rac}
\end{figure}

\section{Beyond-Planar Graphs with Straight-Lines}
\label{sec:straight-lines}
The straight-line setting imposes an additional constraint on the drawings. Therefore, all upper bounds directly apply in the straight-line setting.
Lemma~\ref{lemma:to_rectilinear} shows that it is sufficient to draw each of our lower bound constructions with straight lines to establish the same lower bounds.
We say that two drawings of a graph are \emph{weakly isomorphic} if there is a one-to-one correspondence between their vertices and edges that is incidence-preserving and
crossing-preserving.

\begin{lemma}\label{lemma:to_rectilinear}
Let $G$ be a graph and $\Gamma$ be a drawing of $G$ with $x$ crossings, certifying that $\crnum(G) \leq x$.
Moreover, let
$\mathcal{F}$ be a beyond-planar family of graphs such that $\crnum_{\mathcal{F}}(G) \geq y$.
If there is a straight-line drawing of $G$ weakly isomorphic to $\Gamma$ and a straight-line drawing of $G$ in $\mathcal{F}$, then $\overline{\crnum}(G) \leq x$ and $\overline{\crnum}_{\mathcal{F}}(G) \geq y$.
\end{lemma}
\begin{proof}
Since there is a straight-line drawing of $G$ which has the same number of crossing as $\Gamma$ it immediately follows that $\overline{\crnum}(G) \leq x$.
Drawing graphs with straight lines imposes additional constraints. Since $G$ has a straight-line drawing in $\mathcal{F}$ it follows that $\overline{\crnum}_\mathcal{F}(G) \geq y$.
\end{proof}
In the following, we either describe or directly draw straight-line versions of all drawings required for the lower bound constructions in this paper. This directly implies that all bounds also hold in the straight-line setting.

\mypar{$k$-Planar Graphs.}
We argue that the 1-planar graph $\Gell$ constructed by Chimani~\etal~\cite{chimani} still admits a straight-line drawing when edges are replaced with $k$-planar compound edges.

The graph $\Gell$ without the special edge is planar.
Thus, by Fary's theorem
it admits a plane straight-line drawing.
Moreover, as in Figure~\ref{fig:1-planar-chimani} on the left,
$\Gell$ can be drawn with straight-line edges and such that $\overline{\crnum}(G_\ell) = 2$.
Since $k$-planar compound edges can be drawn with straight-line edges and arbitrarily close to the original edges they replace, we have that $\Gell^k$
admits a straight-line drawing with $2k^2$ crossings.
%
Thomassen showed that for every 1-planar drawing $\Gamma$ there exists a straight-line drawing weakly isomorphic to $\Gamma$ if and only if $\Gamma$ does not contain a $B$- or $W$-configuration~\cite{thomassen1988}. The 1-planar drawing in  Figure~\ref{fig:1-planar-chimani} on the right contains neither a $B$- or a $W$-configuration.
Hence there is a straight-line drawing of $\Gell^k$ that is also $k$-plane.
Thus, by Lemma~\ref{lemma:to_rectilinear}, the bounds of Theorem 1 in~\cite{chimani} and Theorem~\ref{theorem:k-planar} also hold in the straight-line setting.



\mypar{$k$-Quasi-Planar Graphs.}
Chimani~\etal \cite{chimani} showed that $\varrho_{\mathrm{quasi}} \in \Omega(n)$.
Figure~\ref{fig:quasi-straight} is a straight-line version of Figure~2 in~\cite{chimani}, the drawings they use to support this lower bound. Thus, by Lemma~\ref{lemma:to_rectilinear}, the statement of Theorem~4 in~\cite{chimani} also holds in the straight-line setting.
%
Chimani~\etal \cite{chimani} showed that $\varrho_{\mathrm{k-quasi}} \in \Omega(n/k^3)$.
For a different $k$ we can expand the drawings by adding more vertices between $u_2$ and $u_3$ and between $u_4$ and $u_5$ and rearranging the edges such that they all mutually cross, similarly to the argumentation for Corollary 6 in~\cite{chimani}. Again, by Lemma~\ref{lemma:to_rectilinear} the statement of their corollary holds in the straight-line setting.
%

\begin{figure}[h]
    \centering
    \begin{subfigure}{.45\textwidth}
        \centering
        \includegraphics[page=4]{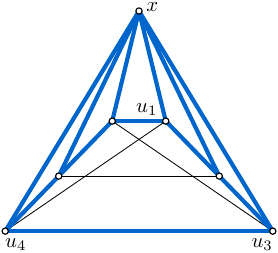}
    \end{subfigure}\hfill
    \begin{subfigure}{.45\textwidth}
        \centering
        \includegraphics[page=5]{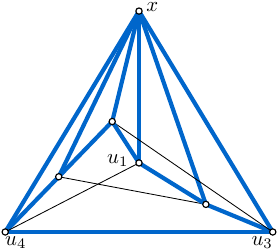}
    \end{subfigure}
    \caption{Straight-line drawings of graph $\Gell$ that are weakly isomorphic to the drawings in Figure~2 in~\cite{chimani}. Left: straight-line drawing of $\Gell$ with $\overline{\crnum}(\Gell) \leq 3$; Right: straight-line quasi-planar drawing of $\Gell$ with $\overline{\crnum}_\quasisub(\Gell) \geq \ell$, $\ell$-compound edges are drawn with thick blue lines.}
    \label{fig:quasi-straight}
\end{figure}

\begin{figure}[b]
    \centering
    \begin{subfigure}{.45\textwidth}
        \centering
        \includegraphics[page=3]{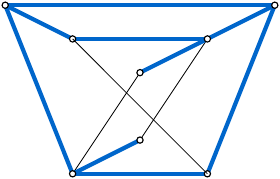}
    \end{subfigure}\hfill
    \begin{subfigure}{.45\textwidth}
        \centering
        \includegraphics[page=4]{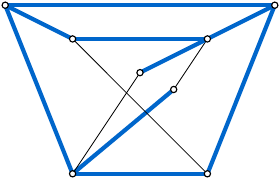}
    \end{subfigure}
    \caption{
    Straight-line drawings of graph $\Gell$ that are weakly isomorphic to the drawings in Figure~3 in~\cite{chimani}. Left: straight-line drawing of $\Gell$ with $\overline{\crnum}(\Gell) \leq 2$; Right: straight-line fan-planar drawing of $\Gell$ with $\overline{\crnum}_\fanplsub(\Gell) \geq \ell$, $\ell$-compound edges are drawn with thick blue lines.}
    \label{fig:fan-straight}
\end{figure}

\mypar{Fan-Planar Graphs.}
Chimani~\etal \cite{chimani} showed that $\varrho_{\mathrm{fan}} \in \Omega(n)$.
Figure~\ref{fig:fan-straight} is a straight-line version of Figure~3 in~\cite{chimani}, the drawings they use to support this lower bound. Thus, by Lemma~\ref{lemma:to_rectilinear}, the statement of Theorem~7 in~\cite{chimani} also holds in the straight-line setting.

\mypar{$(k,l)$-Grid-Free Graphs.}
Figure~\ref{fig:kl-grid-free-straight} is a straight-line version of Figure~\ref{fig:kl-grid-free-appendix}. Thus, by Lemma~\ref{lemma:to_rectilinear}, the statement of Theorem~\ref{theorem:kl} also holds in the straight-line setting for $k=2,l=3$. For a different $k$ we can expand the drawings by adding more vertices between $u_1$ and $u_2$ and between $u_6$ and~$u_7$. For a different $l$ we can do the same between $u_3$ and $u_4$ and between $u_9$ and $u_{10}$.

\begin{figure}[h]
    \centering
    \begin{subfigure}{.45\textwidth}
        \centering
        \includegraphics[page=4]{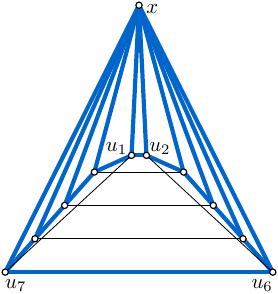}
    \end{subfigure}\hfill
    \begin{subfigure}{.45\textwidth}
        \centering
        \includegraphics[page=5]{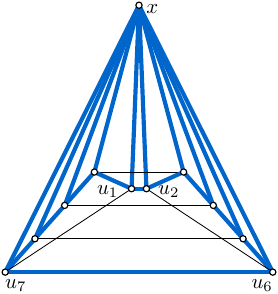}
    \end{subfigure}
    \caption{
    Straight-line drawings of graph $\Gell$ that are weakly isomorphic to the drawings in Figure~\ref{fig:kl-grid-free-appendix}. Left: straight-line drawing of $\Gell$ with $\overline{\crnum}(\Gell) \leq 6$; Right: straight-line 2,3-grid-free drawing of $\Gell$ with $\overline{\crnum}_\klgridsub{2}{3}(\Gell) \geq \ell$, $\ell$-compound edges are drawn with thick blue lines.}
    \label{fig:kl-grid-free-straight}
\end{figure}


\mypar{$k$-Gap-Planar Graphs.}
Figure~\ref{fig:1-gap-straight-lines} is a straight-line version of Figure~\ref{fig:1gap-appendix}. Thus, by Lemma~\ref{lemma:to_rectilinear}, the statement of Theorem~\ref{theorem:1gap} also holds in the straight-line setting.
%
%
For a different $k$ we can expand the drawings by adding more vertices between $u_2$ and $u_3$ and between $u_4$ and $u_5$. Thus, by Lemma~\ref{lemma:to_rectilinear}, the statement of Theorem~\ref{theorem:kgap} also holds in the straight-line setting.
\begin{figure}[h]
    \centering
    \begin{subfigure}{.45\textwidth}
        \centering
        \includegraphics[page=4]{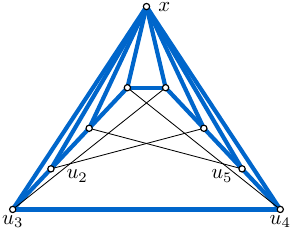}
    \end{subfigure}\hfill
    \begin{subfigure}{.45\textwidth}
        \centering
        \includegraphics[page=5]{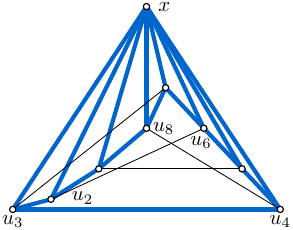}
    \end{subfigure}
    \caption{Straight-line drawings of graph $\Gell$ that are weakly isomorphic to the drawings in Figure~\ref{fig:1gap-appendix}. Left: straight-line drawing of $\Gell$ with $\overline{\crnum}(\Gell) \leq 6$; Right: straight-line 1-gap-planar drawing of $\Gell$ with $\overline{\crnum}_\kgapsub{1}(\Gell) \geq \ell$, $\ell$-compound edges are drawn with thick blue lines.}
    \label{fig:1-gap-straight-lines}
\end{figure}

\mypar{Skewness-$k$ Graphs.}
Figure~\ref{fig:skewness-straight} is a straight-line version of Figure~\ref{fig:1-skewness-appendix}. Thus, by Lemma~\ref{lemma:to_rectilinear}, the statement of Theorem~\ref{theorem:skew} also holds in the straight-line setting.
%
%
For a different $k$ we can expand the drawings by adding more instances of $\Gell$ to~$x$. Thus, by Lemma~\ref{lemma:to_rectilinear}, that the statement of Theorem~\ref{theorem:kskew} also holds in the straight-line setting.
\begin{figure}[h]
    \centering
    \begin{subfigure}{.45\textwidth}
        \centering
        \includegraphics[page=4]{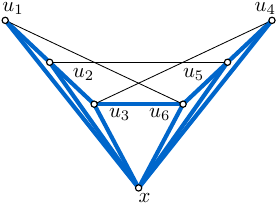}
    \end{subfigure}\hfill
    \begin{subfigure}{.45\textwidth}
        \centering
        \includegraphics[page=5]{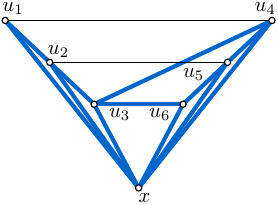}
    \end{subfigure}
    \caption{Straight-line drawings of graph $\Gell$ that are weakly isomorphic to the drawings in Figure~\ref{fig:1-skewness-appendix}. Left: straight-line drawing of $\Gell$ with $\overline{\crnum}(\Gell) \leq 3$; Right: straight-line skewness-1 drawing of $\Gell$ with $\overline{\crnum}_\kskewsub{1}(\Gell) \geq \ell$, $\ell$-compound edges are drawn with thick blue lines.}
    \label{fig:skewness-straight}
\end{figure}

\mypar{$k$-Apex Graphs.}
Figure~\ref{fig:skewness-straight} is a straight-line version of Figure~\ref{fig:1-skewness-appendix}. Thus, by Lemma~\ref{lemma:to_rectilinear}, the statement of Corollary~\ref{corollary:1apex} also holds in the straight-line setting.
For a different $k$ we can expand the drawings by adding more instances of $\Gell$ to $x$. Thus, by Lemma~\ref{lemma:to_rectilinear}, that the statement of Corollary~\ref{corollary:kapex} also holds in the straight-line setting.

\mypar{Planarly Connected Graphs.}
Figure~\ref{fig:planarly-connected-straight} is a straight-line version of Figure~\ref{fig:planarly-connected-appendix}. Thus, by Lemma~\ref{lemma:to_rectilinear}, the statement of Theorem~\ref{theorem:planarcon} also holds in the straight-line setting.


\begin{figure}[h]
    \centering
    \begin{subfigure}{.45\textwidth}
        \centering
        \includegraphics[page=5]{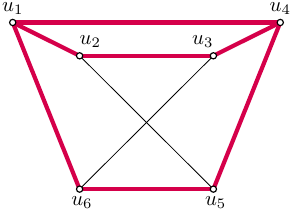}
    \end{subfigure}\hfill
    \begin{subfigure}{.45\textwidth}
        \centering
        \includegraphics[page=6]{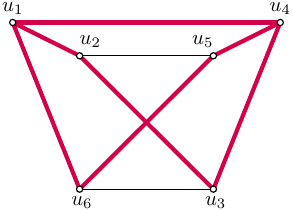}
    \end{subfigure}
    \caption{Straight-line drawings of graph $\Gell$ that are weakly isomorphic to the drawings in Figure~\ref{fig:planarly-connected-appendix}. Left: straight-line drawing of $\Gell$ with $\overline{\crnum}(\Gell) \leq 1$; Right: straight-line planarly connected drawing of $\Gell$ with $\overline{\crnum}_\plconsub(\Gell) \geq \ell^2$, $\ell_2$-bundles are drawn with thick red lines.}
    \label{fig:planarly-connected-straight}
\end{figure}

\mypar{$k$-Fan-Crossing-Free Graphs.}
Figure~\ref{fig:fan-free-straight-lines} is a straight-line version of Figure~\ref{fig:fan-crossing-free-appendix}. Thus, by Lemma~\ref{lemma:to_rectilinear}, the statement of Theorem~\ref{theorem:2fan} also holds in the straight-line setting.
%
For a different $k$ we can expand the drawings in Figure~\ref{fig:fan-free-straight-lines} by making the cycle of length $3+k$, adding the extra vertices between $u_3$ and $u_4$, and adding more edges like $\edge{u_1}{u_3}$ and $\edge{u_1}{u_4}$.

\begin{figure}[h]
    \centering
    \begin{subfigure}{.45\textwidth}
        \centering
        \includegraphics[page=4]{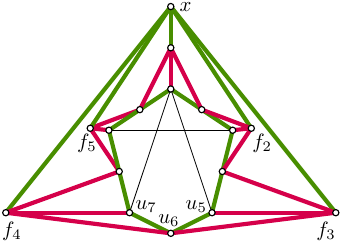}
    \end{subfigure}\hfill
    \begin{subfigure}{.45\textwidth}
        \centering
        \includegraphics[page=5]{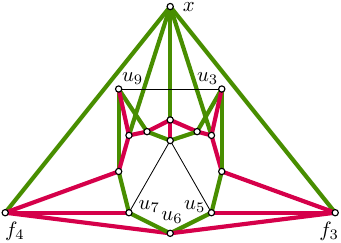}
    \end{subfigure}
    \caption{Straight-line drawings of graph $\Gell$ that are weakly isomorphic to the drawings in Figure~\ref{fig:fan-crossing-free-appendix}. Left: straight-line drawing of $\Gell$ with $\overline{\crnum}(\Gell) \leq 2$; Right: straight-line fan-crossing-free drawing of $\Gell$ with $\overline{\crnum}_\fanfreesub(\Gell) \geq \ell^2$. Thick red and green lines represent
    $\ell_2$- and $\ell_3$-bundles.}
    \label{fig:fan-free-straight-lines}
\end{figure}

\section{Conclusion}

We studied the relation between the crossing number restricted to beyond-planar drawings and the (unrestricted) crossing number and established a number of new lower bounds on the crossing ratio for several classes of beyond-planar graphs. Our results are summarized in Table~\ref{tab:results}. The bounds printed in bold font are tight in the number of vertices $n$. An obvious open question is to improve all bounds which are not tight. Even if a bound is tight in $n$ it might not be tight in $k$, the parameter which describes the graph class in question. In fact, not a single bound is currently tight in $k$ for non-constant $k$.
Furthermore, it might be of interest to investigate what aspects of a beyond-planar family cause a quadratic instead of a linear crossing ratio. Last but not least, the upper bounds generally assume that drawings are simple. We conjecture that this restriction is in fact not needed and that the bounds hold also in the non-simple setting.

\bibliography{journalbiblio}

\end{document}